\documentclass[aip,jmp,amsmath,amssymb,preprint,]{revtex4-1}
\usepackage[utf8]{inputenc}
\usepackage{amsmath}
\usepackage{amsthm}
\usepackage{amssymb}
\usepackage{graphicx}
\usepackage{subfigure}
\usepackage{color}

\newtheorem{lemma}{Lemma}[section]
\newtheorem{theorem}{Theorem}[section]

\begin{document}
	
	\title{Towards Optimal Control of Amyloid Fibrillation}

	\author{Mengshou Wang$^{1}$, Gao Li$^{2,3}$, Liangrong Peng$^{4}$, Liu Hong$^{1}$\footnote{Author to whom correspondence should be addressed. Electronic mail: hongliu@sysu.edu.cn}\\
		$^1$School of Mathematics, Sun Yat-sen University, Guangzhou, 510275, P.R.C.\\
		$^2$Fujian Engineering Research Center of New Chinese Lacquer Material, College of Material and Chemical Engineering, Minjiang University, Fuzhou, 350108, P.R.C\\
		$^3$Institute of Oceanography, Minjiang University, Fuzhou, 350108, P.R.C\\
		$^4$College of Mathematics and Data Science, Minjiang University, Fuzhou, 350108, P.R.C.}
	
	\begin{abstract}
		EGCG, as a representative of amyloid inhibitors, has shown a promising ability against A$\beta$ fibrillation by directly degradating the mature fibrils. Most previous studies have been focusing on its functional mechanisms, meanwhile its optimal dosage has been seldom considered. To solve this critical issue, we refer to the generalized Logistics model for amyloid fibrillation and inhibition and adopt the optimal control theory to balance the effectiveness and cost (or toxicity) of inhibitors. The optimal control trajectory of inhibitors is analytically solved, based on which the influence of model parameters, the difference between the optimal control strategy and several traditional drug dosing strategies are systematically compared and validated through experiments. Our results will shed light on the rational usage of amyloid inhibitors in clinic.
	\end{abstract}
	
	
	\maketitle
	
	\section{Introduction}
	Many human neurodegenerative diseases, including Alzheimer' s disease and Parkinson' s disease, are closely related with amyloid fibrillation \cite{goedert2006century, geschwind2003tau, ueda1993molecular}. For example, the amyloid cascade hypothesis believes that the aggregation of amyloid-$\beta$ (A$\beta$) is a major causative reagent for Alzheimer' s disease. And any effective reduction on the amount of amyloid aggregates will be beneficial to human bodies. However, due to the complexity of amyloid fibrillation processes, especially when considering the variety of amyloid inhibitors, the development of efficient therapeutics for amyloid related diseases is still a difficult task. Till now, all efforts along this direction failed in clinic, except for Aducanumab, a first drug approved by FDA for Alzheimer' s disease in twenty years though full of controversy \cite{alexander2021NEJM}.
	
	To solve this critical issue, in recent years more and more studies have been designed to elucidate the inhibitory mechanisms of different inhibitors. Meanwhile, the potential toxicity or other negative effects of amyloid inhibitors have not received enough attention. Furthermore, some amyloid inhibitors are difficult to synthesize or extract, which makes them quite expensive. Therefore, the exploration on the optimal strategies for the usage of amyloid inhibitors is of great importance in clinic.
	
	In contrast to tremendous applications of optimal control in economics\cite{Kamien2012,Weber2011}, epidemics\cite{Agusto2012,Sharomi2017,shen2021mathematical,lemecha2020optimal}, and cancer\cite{Murray1990,Schaettler2015,Jarrett2020,Cunningham2018,cunningham2020optimal}, \textit{etc.}, there are quite few related studies in the current field. Thomas Michaels \textit{et al.} applied the optimal control theory to analyze the relation between the molecular basis of amyloid fibrillation and optimal regions for inhibition\cite{Michaels2019}. Taking the variability of reaction kinetics into account, Alexander Dear \textit{et al.} used stochastic optimal control theory to determine the optimal dosage of inhibitors that act on several key steps of  amyloid fibrillation\cite{Dear2021}.

	In the current study, we refer to the optimal control theory to design the optimal strategy for the usage of amyloid inhibitors, in particular EGCG -- a small molecule extracted from green tea which shows a strong ability to eliminate mature fibrils \cite{Ehrnhoefer2008EGCG}. Based on the Pontryagin's Maximum Principle (PMP), the optimal strategy governed by a group of ordinary differential equations is derived, which can be explicitly solved when few fibrils are presented in the system. For general situations, phase diagrams revealing the dependence of optimal control strategy on typical dimensionless parameters are numerically explored. The optimal control strategy is further compared with several other traditional strategies, like lump-sum adding, multiple-times adding, \textit{etc.} Their difference is revealed and verified through carefully designed in-vitro experiments.
	
	The whole paper is organized as follows. Sec. I contains the introduction.  The basic kinetic model for amyloid fibrillation and inhibition, the optimal control strategy derived from the PMP, as well as its analytical solution are presented in Sec. II. The parameter dependence of the optimal control strategy and its comparison with other traditional dosing strategies are numerically explored and summarized in Sec. III. The last section is a conclusion. The details of mathematical derivation and experimental setup are left in the appendix.

\section{The Optimal Control of Amyloid Fibrillation}
\subsection{Gerneral Formulation}
	To describe the time-dependent processes of amyloid fibrillation, we denote the interested macroscopic quantities for characterizing the amount of amyloid aggregates by a vector $X(t)\in \mathbb{R}^d$. Among extensive candidates of $X(t)$, the number concentration and mass concentration of total aggregates are two most popular ones. Without loss of generality, the procedure of amyloid fibrillation in the presence of inhibitors is characterized by the following ordinary differential equations (ODEs),
	\begin{equation}\label{dxdt}
		\frac{d X}{d t}=f(X(t),u(t),t),\quad t \in {[0,T]},\quad X(t=0)=x_0,
	\end{equation}
	where the initial time is set to be $0$, $T$ is the terminal time, and \(x_0\) stands for the initial state. Here the concentration of inhibitors $u(t)$ is adopted to describe the regulatory effect of inhibitors. A concrete form of $f$ specifies the detailed kinetic model for amyloid fibrillation and inhibition, \textit{e.g.} the Logistics model, the NES (short for  primary nucleation, elongation, and secondary nucleation) model, \textit{etc.} Interested readers are referred to Refs. \cite{lim2019black,yuan2017biophysics} for further details.
	
	In what follows, we will first consider situations when the terminal time $T$ is fixed while the terminal state \(X(T)\) is free from restrictions. We will then consider the fixed terminal state \(X(T)=0\). By taking the amount of both fibrils and inhibitors into account, our central aim is to minimize the following objective functional $J$ by controlling the time profile of $u(t)$, \textit{i.e.}
	\begin{equation}\label{obj}
		\inf_{u(t)} J[u(t)]=\Phi[X(T)]+\int_{ 0}^{T} L[X(t), u(t), t] d t,
	\end{equation}	
	where $\Phi$ is the terminal cost, $L$ is the running cost. In the filed of optimal control, a model incorporating Eqs.\eqref{dxdt} and \eqref{obj} is known as the Bolza problem, whose solution gives the optimal control strategy.
	
	As an illustration, here we give two particular examples of the above objective functional. When neglecting the running cost and solely trying to minimizing the amount of amyloid aggregates at time $T$, we choose $ \Phi = | | X (T) | |^2 $, $L = 0 $. Here and in what follows we denote the $L_2$-norm by $||\cdot||$.  Contrarily, if neglecting the terminal cost and considering the accumulative effects of amyloid aggregates and inhibitors during the whole process instead, we can set $\Phi = 0 $, $L = | | X (t) | |^2 + | | u (t) | |^2 $.
	
	Besides the objective functional, there are also constraints needed to be considered. For example, the concentrations of fibrils and inhibitors can not be negative ($X(t)\geq0, u(t)\geq0$). If we put further restrictions on the total amount of inhibitors added during the whole procedure, an integral inequality emerges, \textit{i.e.} $\int_0^T u(t)dt\leq u_{tot}$. Therefore, in general we will have two kinds of constraints,
	\[
	C(X(t),u(t),t)=a, \qquad
	S(X(t),u(t),t)\le b,
	\]
	where both functionals $C$ and $S$ are assumed to be continuous with respect to $X(t), u(t)$ and $t$.
	
	As a summary, a general optimal control problem of amyloid fibrillation includes three major parts: the kinetic model, the objective functional and possible constraints. Once they are specified, the optimal control strategy can be derived by referring to the classical optimal control theory.

	\subsection{Kinetics of Amyloid Fibrillation and Inhibition}
	In the above section, we illustrate a general picture for the optimal control problem. Now we proceed to more concrete applications with explicit models for amyloid  fibrillation and inhibition.
	
	Let us denote $M(t)$ and $u(t)$ as the respective mass concentrations of amyloid aggregates and inhibitors at time $t$.  The typical growth of amyloid fibrils follows a sigmodal curve, which means the fibrils grow exponentially at the early stage, then get saturated due to the consumption of free monomers, and finally stop growing. In analogy to population dynamics, this process can be depicted by the Logistics model, \textit{i.e.}
	\begin{equation}\label{equ:log}
		\frac{d}{d t}M(t)=k_{a}M(t)\left(1-\frac{M(t)}{M_{max}}\right),
	\end{equation}
	where $k_{a}$ is the apparent fibril growth rate, $M_{max}$ is the maximal mass concentration of aggregates.
	
	\begin{figure}
		\centering
		\includegraphics[width=6.5in]{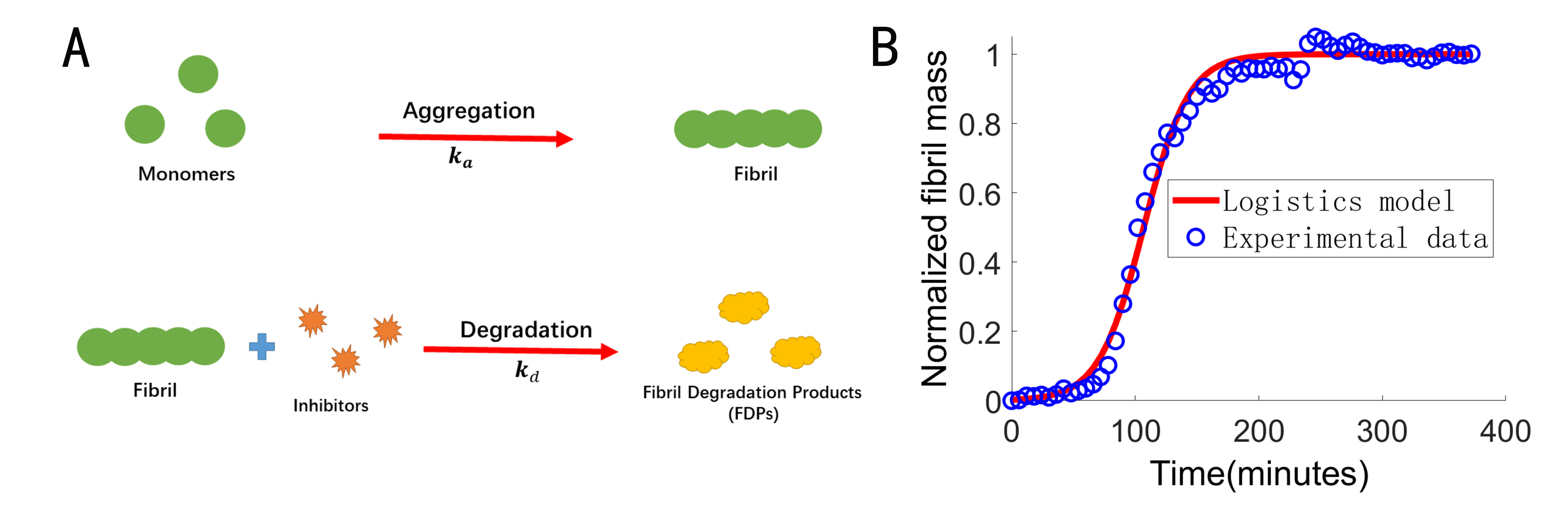}
		\centering
		\caption{(A) Sketch map of A$\beta$ fibrillation and inhibition by EGCG, and (B) kinetics of fibril growth.}
		\label{figure_1}
	\end{figure}
	
	In the presence of inhibitors, the conventional Logistics model requires further modification. For many natural compounds against amyloids, especially those against A$\beta$ such as sulfated polysaccharides from the sea cucumber \cite{li2021Inhibitory} and EGCG from green tea \cite{Ehrnhoefer2008EGCG}, they show strong abilities in degrading  amyloid aggregates directly, which means the rate of inhibition is proportional to the product of concentrations of inhibitors and fibrils as shown in Fig. 1(A). Therefore, we have a generalized Logistics model as
	\begin{equation}
		\label{equ:state}
		\frac{d}{d t}M(t)=k_{a}M(t)\left(1-\frac{M(t)}{M_{max}}\right)-k_{d} u(t) M(t),
	\end{equation}
	where \(k_{d}\) is the rate constant of fibril inhibition.
	
	It is straightforward to see that the mass concentrations of fibrils and inhibitors cannot be minimized simultaneously.
	In order to decrease these two concentrations as much as possible during the whole process (for the purpose of reducing the toxicity or cost), we need to make a trade-off.
	For simplicity, here we choose a simple quadratic objective functional,
	\begin{equation}\label{equ:obj}
		J[u(\cdot)]=\int_{0}^{T}[M(t)^2+w^2u(t)^2]d t,
	\end{equation}
	with $w>0$ denoting the relative weight. The objective functional with a terminal cost,
	\begin{equation}\label{equ:obj2}
		J[u(\cdot)]=vM(T)+\int_{0}^{T}[M(t)^2+w^2u(t)^2]d t,
	\end{equation}
	with $v>0$ denoting the relative weight,
	can be studied analogously and will not be addressed here.
	Here the $L^2$ norm is chosen only for simplicity. Alternative choice, like the $L^1$ norm, is given in Appendix \ref{L1}. The connection and difference between $L^{1}$ and $L^2$ norms have been discussed by Schattler \textit{et al.} during the application of optimal control in oncology \cite{schattler2015optimal}.

The mass concentrations of amyloid aggregates and inhibitors need to be non-negative at any time point, that is to say,
		\begin{equation}\label{constraints}
		M_{max} \geq M(t) \geq 0 \quad \text{ and} \quad u_{max}\geq u(t) \geq 0,
        \qquad \forall t \geq 0.
		\end{equation}
In fact, the existence of optimal control subjected to these inequalities can be given.
For $u_{max}>k_{d}M_{max}^2T/w^2$ (which will be clear),
the existence and uniqueness of the optimal control problem is proved, whose solution is  equivalent to solving a boundary value differential equation.
	
	\subsection{Optimal Control Strategy}
	Towards our problem, we can obtain the following important theorem by conducting  strict mathematical derivations.
	\begin{theorem}\label{main theorem1}
		Consider the optimal control of amyloid fibrillation given by the generalized Logistics model in Eq.\eqref{equ:state}, the quadratic cost in Eq.\eqref{equ:obj} and the constraints in Eq.\eqref{constraints}. Further suppose the terminal state \(X(T)\) to be free. When $u_{max} > k_{d}M_{max}^2T/w^2$, there exists a unique optimal control strategy, which is given by
		\begin{equation}\label{PMP_optimal solution}
			\left\{\begin{array}{cl}
				\dot{M^*}(t)=&k_{a}M^*(t)(1-\frac{M^*(t)}{M_{max}})-k_{d} u^*(t) M^*(t),\\
				\dot{u^*}(t)=&\frac{k_{a}}{M_{max}}u^*(t)M^*(t) - \frac{k_{d}}{w^2}{M^*(t)}^{2},~t< T,
				\\M^*(0)=&{m_0}, \quad u^* (T)=0.
			\end{array}\right.
		\end{equation}
	where $u^{*}(t)$ denotes the optimal control strategy, and $M^{*}(t)$ is the corresponding optimal state.
	\end{theorem}
	\begin{proof}
	We put the proof in Appendix \ref{Ap_U-C}.
	\end{proof}


	\subsection{Analytical Solutions When Fibrils are Few}
	\label{approximate analytic solution}
	
	Generally speaking, the optimal control strategy given in Eq. \eqref{PMP_optimal solution} has to be solved numerically, \textit{e.g.} by using the shooting method. However, when the mass concentration of fibrils is relatively low, it can be solved analytically under proper approximations.
	
	Notice that Eq. (\ref{equ:state}) can be linearized around its unstable fixed point \(M=0\), which corresponds to a low concentration of fibrils.
	This situation is practically significant when considering the suppressing effect of inhibitors.
	Given $M(t)/M_{max}\ll1$, the right-hand side of Eq.\eqref{equ:state} becomes
	\begin{equation}\label{equ:stateapp}
		\dot{M}(t)=k_{a}M(t)-k_{d} u(t) M(t),\quad M(0)={m_0},\quad t \in (0,T],
	\end{equation}
	where \({m_0}\) is a small initial value. Meanwhile, $u^*(t)$ satisfies a second-order ODE (see Appendix \ref{Ap_A} ), which reads
	\begin{equation}\label{u^*2}
		\left\{\begin{array}{cl}
			\ddot{u}^*(t)=2k_{a}\dot{u}^*(t)-2k_{d}u^*\dot{u}^*(t),\\
			u^*(T)=0,\quad \dot{u}^{*}(0)=-k_{d}{m_{0}}^{2}/w^2.
		\end{array}\right.
	\end{equation}
		In this case, the optimal trajectory \(M^*(t)\) and the optimal control strategy \(u^*(t)\) can be analytically solved,
	\begin{equation}\label{an_solu}
		u^*(t)=C_1\tan(-k_{d}C_1t+C_2)+\frac{k_{a}}{k_{d}},\quad
		M^*(t)=wC_1\sec(-k_{d}C_1t+C_2),
	\end{equation}
	where the constants $C_1$ and $C_2$ are given through the equalities $ {m_0}^2=w^2C_1^2\sec^2(C_2)$ and $ {k_{a}^2}/{k_{d}^2}=C_1^2\tan^2(-k_{d}C_1T+C_2)$. 	Details of calculation are given in Appendix \ref{Ap_A}.

	With respect to the above analytical solutions, the following conclusions can be reached.
	
	(i) If both $w$ and $m_{0}$ are proportionally enlarged by $k$ times, then $u^*(t)$ remains invariant, while $M^*(t)$ is also proportionally enlarged by $k$ times.
	
	(ii) If both $w$ and $k_{d}$ are proportionally enlarged by $k$ times, then $M^*(t)$ remains unchanged, while $u^*(t)-k_{a}/k_{d}$ is reduced by $k$ times.
	
	(iii) If both $k_{a}$ and $k_{d}$ are proportionally enlarged by $k$ times and $T$ is reduced by $k$ times, then $M^*(T)$ remains invariant.
%
%
	
	Intuitively, above linear approximation holds as long as the mass concentration of fibril remains low. To further test the validity of this hypothesis, we compare the approximate optimal trajectory of state $M^*(t)$ and the optimal strategy of inhibitors $u^*(t)$ in Eq.\eqref{an_solu} with the corresponding exact solutions in Eq.\eqref{PMP_optimal solution}. In Fig. \ref{figure_2}, it is observed that with the increment of \( m_{0}{k_{d}}/{k_{a}}\), a dimensionless parameter representing the degradation effect of inhibitors against amyloid fibrillation, the approximate solutions converge to the exact values rapidly. And its validity holds in a region even far beyond the original hypothesis that the fibrils should be kept relatively few.
	
	\begin{figure}[htbp]
		\centering
		\includegraphics[width=6in]{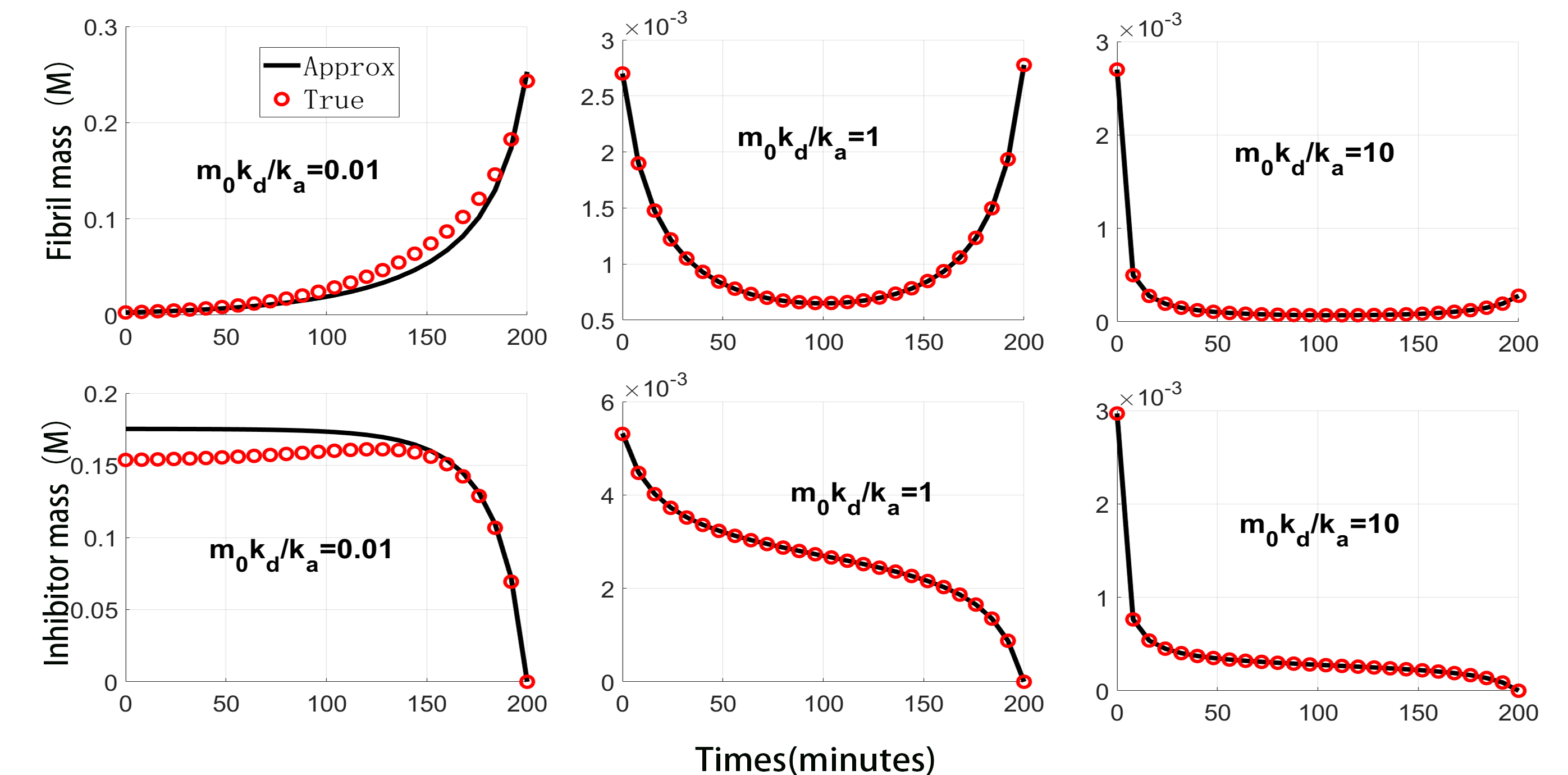}
		\centering
		\caption{Comparison of the approximate solutions (black lines) with numerical solutions (red dots) on both the optimal fibril concentration $M^*(t)$ (upper  panels) and inhibitor concentration $u^*(t)$ (lower panels). From left to right, the dimensionless parameter \(m_0{k_{d}}/{k_{a}}\)  varies as \(0.01,1,10\), while the rest parameters are kept as \(m_0=2.7 \times 10^{-3}\)\(\mathbf{\mu M}\), \(M_{max}=1\)\(\mathbf{\mu M}\), \(k_{a}=5.6 \times 10^{-2}\)\(\min^{-1}\), \(T=200 \min \) and \(w^2=1\).}
		\label{figure_2}
	\end{figure}
	
	\section{Phase Diagram and Strategy Comparison}
	
	In this section, we proceed to make a more thorough exploration on the optimal control strategy based on numerical solutions. A major difficulty is the optimal control strategy $u^*(t)$ follows a backward equation, which has to be solved from time $t=T$ to $t=0$. Here we adopt the shooting method, which is a standard approach to find the correct starting point of $u^*(t)$ at time $t=0$. Thanks to the fact that the terminal state \(u^*(T)\) is monotonous with respect to the initial state \(u(0)\), so that we can use the method of dichotomy to make an efficient guess.

	\subsection{Phase Diagram for the Optimal Control Strategy} 
	
	To be specific, we focus on the optimal control problem of using EGCG against A\(\beta\) aggregation. The apparent fibril growth rate \(k_{a}\) is determined through the fibrillation data of A\(\beta\) by performing nonlinear fitting with respect to the Logistic model (see Fig. \ref{figure_1}(B)), which gives \(k_{a}=5.6 \times 10^{-2} \min^{-1}\) under the condition \(M_{max}=1\)\(\mathbf{\mu M}\), \(m_0=2.7 \times 10^{-3}\)\(\mathbf{\mu M}\).
	
	To elucidate the inhibitory strength of inhibitors on fibril growth and the tolerance of inhibitors during the medical treatment, two dimensionless parameters, \( M_{max}{k_{d}}/{k_{a}}\) and \(w^2\), are found to play a key role in the optimal control problem. As shown in Figs. \ref{figure_3}(A)-(B), by fixing the inhibitor strength $M_{max}{k_{d}}/{k_{a}}$, the lower the tolerance of inhibitors $w^2$ is, the less apparent the inhibitory effect will be. In contrast, the amount of totally added inhibitors shows a non-monotonic dependence on the inhibitor's strength when $w^2$ is fixed. Therefore, a trade-off between the inhibitory strength and the amount of added inhibitors is presented in the optimal control.
	
	\begin{figure}[htbp]
		\centering
		\includegraphics[width=5in]{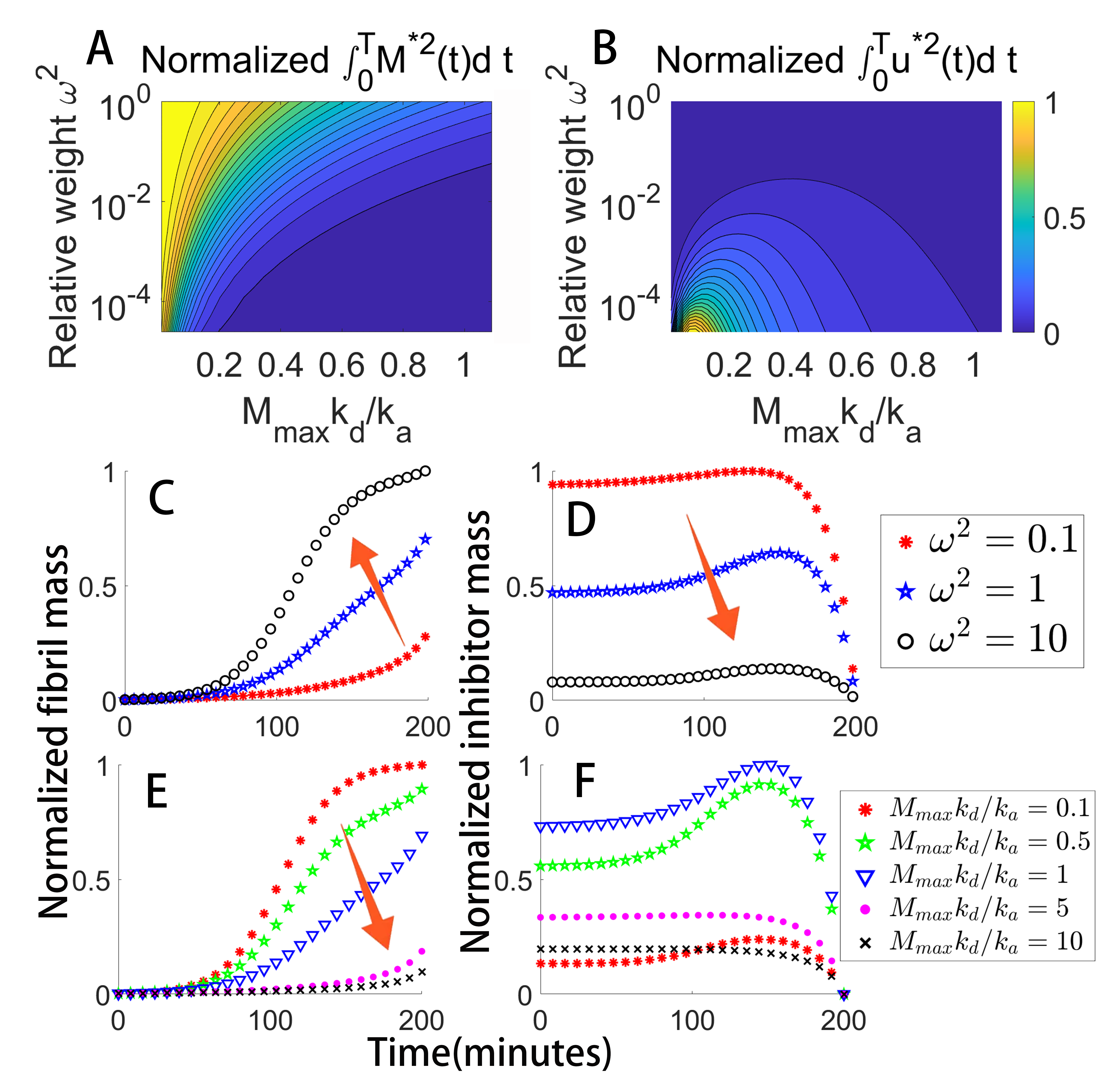}
		\centering
		\caption{Phase diagrams for the accumulative mass concentrations of (A) fibrils \(\int^{T}_{0}{M^*}^{2}(t)dt\) and (B) inhibitors \(\int^{T}_{0}{u^*}^{2}(t)dt\), as a function of the relative weight $w^2$ and $M_{max}k_d/k_a$. For clearance, here \(\int^{T}_{0}{M^*}^{2}(t)dt\) and \(\int^{T}_{0}{u^*}^{2}(t)dt\)  are normalized by their respective maximal values in the simulated region. The optimal trajectories of fibrils and inhibitors are shown respectively, under the conditions of (C, D) \(M_{max}k_d/k_a =1, w^2=0.1,1,10\) and (E, F) \(M_{max}k_d/k_a =0.1,0.5,1,5,10, w^2=1\).}
		\label{figure_3}
	\end{figure}
	
	Next, we look into the kinetic trajectories of $M^*(t)$ and $u^*(t)$ to explore the influence of parameters on the details of optimal control. As shown in Figs. \ref{figure_3}(C)-(D), as the contribution of inhibitor cost becomes more and more significant to the objective functional (meaning larger $w^2$), the fibril mass concentration gets higher and higher in the optimal state, meanwhile the inhibitor concentration shows an opposite tendency as expected. The influence of \( M_{max}{k_{d}}/{k_{a}}\) on the optimal control is a bit complicated. Overall, by increasing \( M_{max}{k_{d}}/{k_{a}}\), the fibril mass concentration drops quickly, demonstrating the improvement of inhibitory effect against the aggregation processes. Similar to the phase diagram in Fig. \ref{figure_3}(B), the inhibitor concentration shows a non-monotonic dependence on \( M_{max}{k_{d}}/{k_{a}}\). And the highest inhibitor concentration is achieved when the fibrillation and inhibition processes are of the same amplitude \( M_{max}{k_{d}}/{k_{a}}=1\). These typical behaviors of the optimal control trajectories are found in Figs. \ref{figure_3}(E)-(F).
	
In most cases, we are interested in the efficient inhibitors, which means $M_{max}{k_{d}}/{k_{a}}\gg1$. Under this condition, it is found that with the increase of the apparent fibril growth rate \(k_{a}\) or the tolerance of inhibitors $w^2$, more inhibitors are required from the early time to reach an optimal control. Meanwhile, opposite tendency is observed for the inhibition rate \(k_{d}\) and the initial monomer concentration \(m_{0}\) as shown in Figs. \ref{figure_6}(A-D). The corresponding optimal trajectories for fibrils and their dependence on model parameters are presented in Figs. \ref{figure_6}(E-H).
	\begin{figure}[htbp]
		\centering
		\includegraphics[width=5in]{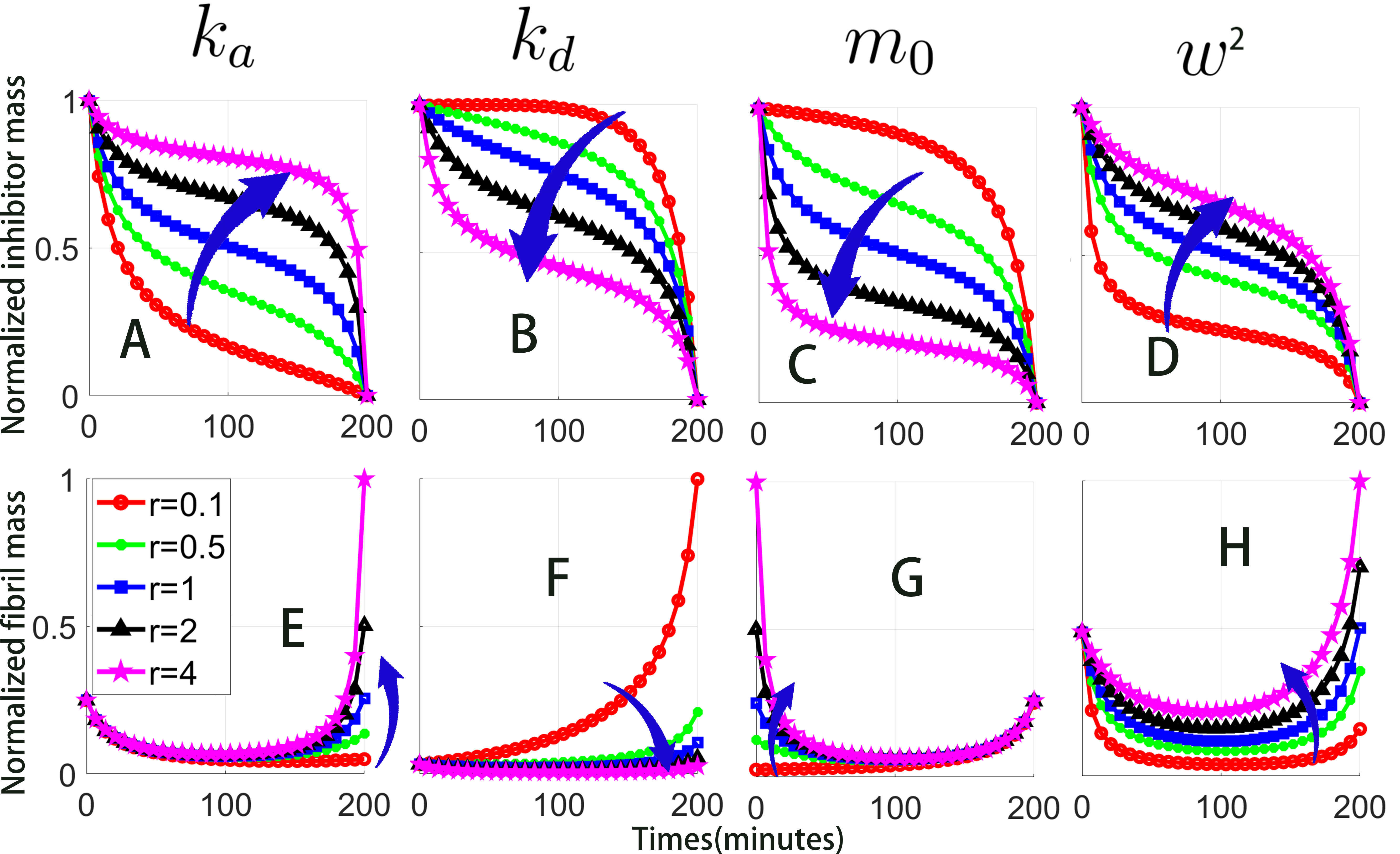}
		\centering
		\caption{Influence of model parameters on the optimal trajectories of inhibitors and fibrils. In each subplot, only the remarked parameter is changed by $r$ times $(r=0.1, 0.5, 1, 2, 4)$, while the rest parameters are kept unchanged. Default parameters are set as \(m_0=2.7 \times 10^{-3}\)\(\mathbf{\mu M}\), \(M_{max}=1\)\(\mathbf{\mu M}\), \(k_{a}=5.6 \times 10^{-2}\)\(\min^{-1}\), \(k_d=k_a/m_0\) and \(w^2=1\).}\label{figure_6}
	\end{figure}

	\subsection{Comparison with Traditional Control Strategies}
	
	Referring to the PMP, the optimal control strategy is undoubtedly the one which leads to the minimal objective functional. However, the optimal control strategy usually follows a complicated trajectory, which may bring great troubles in real applications. In this section, we aim to compare the performance of the optimal control strategy with several other traditional strategies and seek for alternative simple control strategies  with acceptable performance.
	
	Several commonly adopted drug dosing strategies are summarized here, including the \emph{lump-sum adding}, meaning adding drugs all at once; \emph{constant adding}, meaning continuously adding drugs at a constant speed; \emph{multiple adding}, meaning adding drugs instantly with equal amount for multiple times; and \emph{periodic adding}, meaning adding drugs for several time periods and for each period drugs are added constantly. Details are listed in Appendix \ref{Ap_B}.
	
	
	
	Intuitively, when restricted to the fibril growth stage, the \emph{lump-sum adding} is a good strategy because it inhibits fibril growth from the beginning. In contrast, the \emph{multiple adding} is preferred once the equilibrium state is reached, since it can effectively reduce the running cost of inhibitors. As demonstrated through the relative difference between the \emph{lump-sum adding} and \emph{multiple adding} (including \emph{twice adding} and \emph{four times adding}) to the \emph{optimal strategy} in Figs. \ref{figure_5}(A) and (B), we can clearly see that \emph{the lump-sum adding} is better as long as \(T\leq1.4 t_{1/2}\). Figs. \ref{figure_5}(E) and (F) further verify our idea that the \emph{multiple adding} becomes better than the \emph{lump-sum adding} after reaching the equilibrium state. However, there is no significant difference in the fibril mass concentration between the two cases as shown in Figs. \ref{figure_5}(C)-(D).

	\begin{figure}[htbp]
		\centering
		\includegraphics[width=4in]{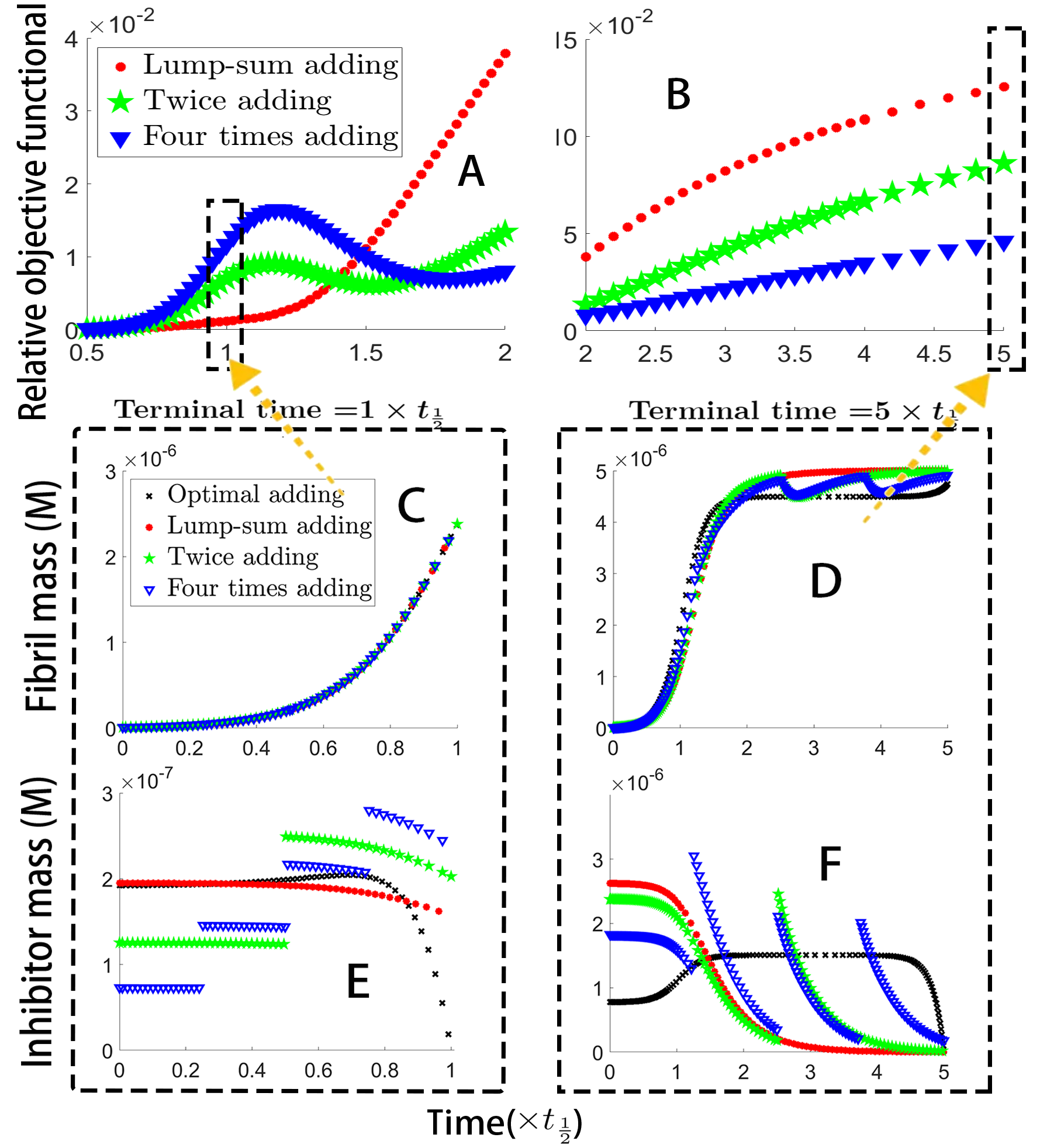}
		\centering
		\caption{Dependence of (A-B) the relative objective functional on the terminal time for the lump-sum adding and multiple-time adding (including {twice adding} and {four times adding}) strategies. Representative optimal trajectories for (C-D) fibril mass concentration and (E-F) inhibitor mass concentration with the different terminal time.}
		\label{figure_5}
	\end{figure}
	
	
	For the purpose of validation, A$\beta$40 control experiments are carefully designed. EGCG, a small molecule extracted from green tea and showing a strong ability in eliminating mature fibrils \cite{Ehrnhoefer2008EGCG}, is adopted for A$\beta$40 inhibition. Five controlled groups, including one \emph{lump-sum adding}, two \emph{twice adding} and two \emph{four-times adding}, have been conducted in Fig. \ref{figure_4}(A) and recorded in Fig. \ref{figure_4}(C).
	
	\begin{figure}[htbp]
		\centering
		\includegraphics[width=5in]{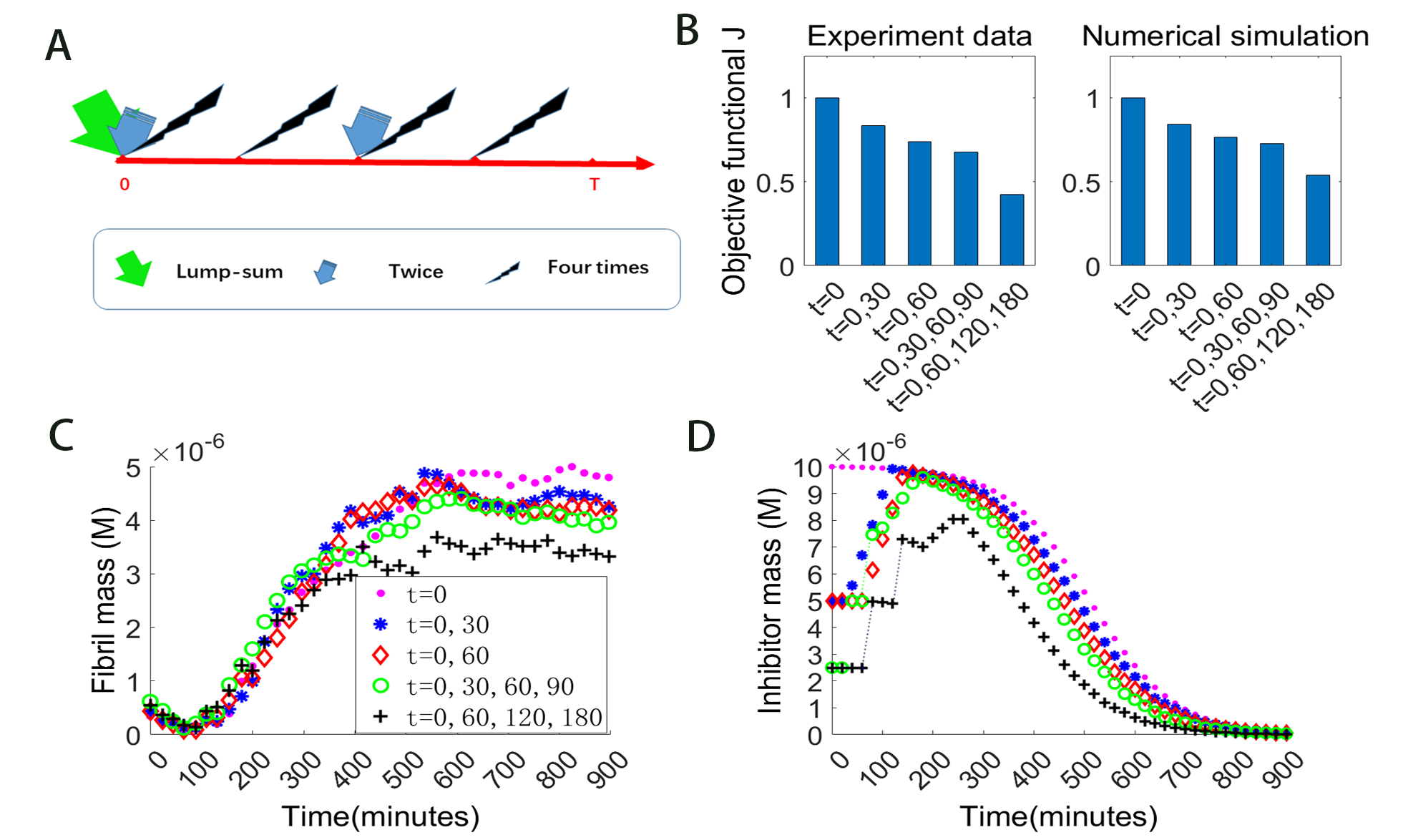}
		\centering
		\caption{(A) Illustration of lump-sum, twice adding and four-time adding strategies. (B) The objective functional $J[u(t)]$ for different drug-dosing  strategies calculated based on either experiment data (left panel) or theoretical predictions (right panel). The adding time is marked below. The corresponding trajectories of (C) fibril mass concentration measured in experiment and (D) inhibitor mass concentration predicted by the optimal control theory are shown respectively. }
		\label{figure_4}
	\end{figure}
	
	\section{Conclusion}
	In this paper, a general theory for the optimal control problem of amyloid aggregation is formulated. With respect to the modified Logistics model with degradation caused by inhibitors, and a general quadratic objective functional including contributions from both fibrils and inhibitors, the optimal control problem is transformed into a group of backward ordinary differential equations by the PMP. It is then solved either analytically under linear approximation or numerically by the shooting method. The influence of non-dimensional parameters, \({M_{max}k_{d}}/{k_{a}}\) representing the strength of inhibitors against fibrillation and \(w^2\) representing the tolerance of inhibitors, on the optimal control trajectories is explored in detail. Finally, several commonly used drug dosing strategies are numerically compared with the theoretically predicted optimal control strategy. We find that for a long-term disease treatment, multiple-time additions are better, while for a short-term treatment, single high-dose treatment is preferred. Our conclusion is further validated through carefully designed experiments of EGCG against A$\beta$ fibrillation.
	
	It should be noted that our current model is an over simplification of the real biochemical kinetics, not only for the amyloid fibrillation but also for the inhibitor acting processes. More detailed kinetic models, especially those based on microscopic mechanisms, need to be further considered in order to provide a more reasonable physical picture. Besides, other combinations of amyloid proteins and inhibitors with different mechanisms are worthy of a systematic examination. It would help to answer the general problem that how the optimal control strategy varies with the system setup which is crucial for clinic practice. We hope our work can raise further interest on the optimal usage of amyloid inhibitors.

	\section*{Appendix}
\setcounter{equation}{0}
\renewcommand\theequation{A.\arabic{equation}}
	\subsection{Proof of  main theorems}\label{Ap_U-C}
	\subsubsection{Some related lemmas on optimal control}
	The Bolza problem of optimal control reads
	\begin{equation}\label{opti_prob}
		\begin{array}{l}
			\inf_{u(t)} J[u(t)]=\Phi\left[X(T)\right]+\int_{0}^{T} L\left[X\left(t\right), u\left(t\right), t\right] d t  \\
			\text{subject to} \\
			\dot{X}(t)=f(X(t),u(t),t),\quad t \in {[0,T]},\quad X(0)=x_0.
		\end{array}
	\end{equation}
	The \emph{Pontryagin's Maximum Principle} (PMP) provides a necessary condition, whose concrete form is stated as follows.
	\begin{lemma}\label{PMP}\cite{pontryagin1987mathematical}
		(PMP) Let \(u^*(t) \in \Theta \subset \mathbb{R}^{m}\) be a bounded, measurable and admissible control that optimizes Eq.(\ref{opti_prob}), with \(\Theta\) be the control set, and \(X^*\) be its corresponding state trajectory. Define a Hamiltonian
		\[H(X,u,\lambda,t)=\lambda^T f(X,u,t)-L(X,u,t),\]
		where $\lambda \in \mathbb{R}^d$. Then there exists an absolutely continuous process \(\lambda(t)\) such that
		\begin{eqnarray}
			&&\dot{X^*}(t)=\frac{\partial H(X^*(t),u^*(t),\lambda^*(t),t)}{\partial \lambda},\quad X^*(0)=x_0
			\label{stateeq}
			\\
			&&\dot{\lambda}^*(t)=-\frac{\partial H(X^*(t),u^*(t),\lambda^*(t),t)}{\partial X},\quad \lambda^*(T)=-\frac{\partial \Phi(X^*(T))}{\partial X}
			\label{costateeq}
			\\
			&&H(X^*(t),u^*(t),\lambda^*(t),t) \ge H(X^*(t),u(t),\lambda^*(t),t),\quad \nonumber\\
			&&\forall u\in \Theta \quad \text{and} \quad a.e. ~t\in [0,T]\nonumber
		\end{eqnarray}
	\end{lemma}
	\noindent Here Eq.\eqref{stateeq} reduces to the state equation under the optimal  control, while the co-state $\lambda^*$ evolves backward according to Eq.\eqref{costateeq} with a fixed  terminal state $\lambda^*(T)$.
	
	\begin{lemma}\label{comparison theorem}\cite{michel1901maniere}
		(Comparison Theorem) Let both functions $f(x,y)$ and $F(x,y)$ be continuous in the plane region $G$ and satisfy the inequality \[f(x,y)<F(x,y),\qquad (x,y) \in G.\] Let $y=\phi(x)$ and $y=\Phi(x)$ be solutions to the initial value problems \[\frac{d y}{d x}=f(x,y),\qquad y(x_{0})=y_{0}\] and \[\frac{d y}{d x}=F(x,y),\qquad y(x_{0})=y_{0}\] respectively on the interval $a<x<b$, where $(x_{0},y_{0})\in G$. Then we have \[
		\begin{cases}
			\phi(x)<\Phi(x),&\qquad x_{0}<x<b,\\
			\phi(x)>\Phi(x),&\qquad a<x<x_{0}.
		\end{cases}
		\]
	\end{lemma}
	
	\begin{lemma}\label{existence uniqueness of bvp}\cite{waltman1971existence}
		 Consider the boundary value problem
		 \begin{eqnarray}
		 	&&\frac{d x}{d t}=f_{1}(t,x,y), \qquad \frac{d y}{d t}=f_{2}(t,x,y),\label{bvp1}\\
		 	&&Ax(a)+By(a)=c_{1},\label{bvp2}\\
		 	&&Cx(b)+Dy(b)=c_{2}, \qquad b>a,\label{bvp3}
		 \end{eqnarray}
	 where the $f_{i}(t,x,y)$ satisfy a Lipschitz condition in the form of
	 \begin{equation}
	 	\begin{aligned}
	 		L_{1}(x_{1}-x_{2}) \leq f_{1}(t,x_{1},y)-f_{1}(t,x_{2},y)\leq L_{2}(x_{1}-x_{2}),\quad \mathtt{if} \quad x_{1} \ge x_{2},\\
	 		K_{1}(y_{1}-y_{2}) \leq f_{1}(t,x,y_{1})-f_{1}(t,x,y_{2})\leq K_{2}(y_{1}-y_{2}),\quad \mathtt{if} \quad y_{1} \ge y_{2},\\
	 		M_{1}(x_{1}-x_{2}) \leq f_{2}(t,x_{1},y)-f_{2}(t,x_{2},y)\leq M_{2}(x_{1}-x_{2}),\quad \mathtt{if} \quad x_{1} \ge x_{2},\\
	 		N_{1}(y_{1}-y_{2}) \leq f_{2}(t,x,y_{1})-f_{2}(t,x,y_{2})\leq N_{2}(y_{1}-y_{2}),\quad \mathtt{if} \quad y_{1} \ge y_{2}.
	 	\end{aligned}
	 \end{equation}
 	And define \[
 	\begin{split}
 		P_{1}(u,v)&=L_{1}u+K_{1}v \quad \mathtt{if} \quad uv \ge 0,\\
 		&=L_{2}u+K_{1}v \quad \mathtt{if} \quad uv < 0;\\
 		P_{2}(u,v)&=M_{2}u+N_{2}v \quad \mathtt{if} \quad uv \ge 0,\\
 		&=M_{2}u+N_{1}v \quad \mathtt{if} \quad uv < 0;\\
 		Q_{1}(u,v)&=L_{2}u+K_{2}v \quad \mathtt{if} \quad uv \ge 0,\\
 		&=L_{1}u+K_{2}v \quad \mathtt{if} \quad uv < 0;\\
 		Q_{2}(u,v)&=M_{1}u+N_{1}v \quad \mathtt{if} \quad uv \ge 0,\\
 		&=M_{1}u+N_{2}v \quad \mathtt{if} \quad uv < 0.\\
 	\end{split}
\]
The comparison system will be
	\begin{equation}\label{eee}
		\frac{d u}{d t}=P_{1}(u,v),\qquad \frac{d v}{d t}=P_{2}(u,v)
	\end{equation}
and
	\begin{equation}\label{eeeee}
		\frac{d w}{d t}=Q_{1}(w,z),\qquad \frac{d z}{d t}=Q_{2}(w,z).
	\end{equation}
	Define $p(\alpha,\beta)$ to be the time it takes for a solution of \eqref{eee} which begins at $t=0$ at a point on $y=\alpha x$ to first reach a point on the line $y=\beta x$ (if a solution does reach it), or to be $+\infty$ if no solution beginning on $y=\alpha x$ at $t=0$ reaches $y=\beta x$ for any $t \ge 0$. Define $q(\alpha,\beta)$ in the same way for the system \eqref{eeeee}.
	
	If $b-a<M$, where
	\[M=\min\left[p(\tan^{-1}-A/B,\tan^{-1}-C/D),q(\tan^{-1}-A/B,\tan^{-1}-C/D)\right],\]
	then there exists a unique solution to the boundary value problem in Eqs.\eqref{bvp1}, \eqref{bvp2} and \eqref{bvp3}.
	\end{lemma}
	We prove that the following boundary value problem exists a unique solution.
	\begin{theorem}\label{bvpeu}
		For $t \in [0,T]$, the boundary value problem,
		\begin{equation}\label{bvppp}
			\left\{\begin{array}{cl}
				\dot{M}(t)=&k_{a}M(t)(1-\frac{M(t)}{M_{max}})-k_{d} u(t) M(t),\\
				\dot{u}(t)=&\frac{k_{a}}{M_{max}}u(t)M(t) - \frac{k_{d}}{w^2}{M(t)}^{2},~t< T,
				\\M(0)=&{m_0}>0, \quad u(T)=0,
			\end{array}\right.
		\end{equation}
		exists a unique solution.
	\end{theorem}
	\begin{proof}
		Our main idea is to use Lemma \ref{existence uniqueness of bvp}. By using Lemma \ref{comparison theorem}, Eq.\eqref{bvppp} implies that $0 \le M(t)<M_{max}$ and $0 \le u(t)<k_{d}M_{max}^2T/w^2$.
		Applying Lemma \ref{existence uniqueness of bvp} to Eq.\eqref{bvppp} with $a=0, b=T, A=1, B=0, c_{1}=m_{0}, C=0, D=1 \text{ and } c_{2}=0$, we have  $\tan^{-1}(-A/B)=-\pi/2$, $\tan^{-1}(-C/D)=0$. We can choose $L_{1}=-k_a-(k_{d}M_{max})^2T/w^2$, $L_{2}=k_{a}$, $K_{1}=0$, $K_{2}=k_{d}M_{max}$, $M_{1}=0$, $M_{2}=-2k_{d}M_{max}$, $N_{1}=k_{a}$ and $N_{2}=0$ in the above Lipschitz condition. Considering Eq.\eqref{eee} and Eq.\eqref{eeeee}, if $(u(0), v(0))=(1, -\pi/2)$, then $(u(t), v(t)) \neq (1, 0)$, $\forall t>0$. So we have $p(-\pi/2,0)=+\infty$ and $q(-\pi/2,0)=+\infty$. The desired result follows.
	\end{proof}
	
	\subsubsection{Proof of theorem \ref{main theorem1} }
	If our control domain is a compact set, the existence of optimal control problems \cite{kelly2016impact} where the terminal state \(X(T)\) is free from restrictions is obvious.
	
	The optimal control problem satisfying the constraint \eqref{constraints} is given by:
	\begin{equation}\label{OC_c}
		\begin{array}{l}
			\inf_{u(t)\in [0,u_{max}]} J[u(\cdot)]=\int_{0}^{T}[M(t)^2+w^2u(t)^2]d t  \\
			\text{subject to} \\
			\dot{M(t)}=k_{a}M(t)\left(1-\frac{M(t)}{M_{max}}\right)-k_{d} u(t) M(t),\quad M^*(0)={m_0}>0,
		\end{array}
	\end{equation}
	where  $M_{max} \geq M(t) \geq 0, \forall t \in [0,T]$.
	
	According to Lemma \ref{PMP}, an optimal solution to \eqref{OC_c} should satisfy the following equations,
	\begin{eqnarray}
		&&\dot{M^*}(t)=k_{a}M^*-k_{a}(M^{*})^2/M_{max}-k_{d} u^* M^*,\label{PMP1}\\
		&&\dot{\lambda^*}(t)=-k_{a}\lambda^*+2k_{a}\lambda^*M^*/M_{max}+k_{d} u^*\lambda^* +2M^*,~t<T, \label{PMP2}\\
		&&u^*(t) \in {argmax}_{u\in [0,u_{max}]}
		\left[\lambda \left(k_{a}M^*-k_{a}(M^*)^2/M_{max}-k_{d} u M^*\right)-(M^*)^2-w^2u^2 \right], \label{PMP3}\\
		&&M^*(0)={m_0}, \qquad \lambda^* (T)=0.\label{PMP4}
	\end{eqnarray}
	If we do not constrain $u^{*}(t)$, we get $u^*(t)=-k_{d}\lambda^*M^*/2w^2$ by Eq.\eqref{PMP3}. Further, we can get
	\begin{equation}\label{fzeq}
		\left\{\begin{array}{cl}
			\dot{M^*}(t)=&k_{a}M^*(t)(1-\frac{M^*(t)}{M_{max}})-k_{d} u^*(t) M^*(t),\\
			\dot{u^*}(t)=&\frac{k_{a}}{M_{max}}u^*(t)M^*(t) - \frac{k_{d}}{w^2}{M^*(t)}^{2},~t< T,
			\\M^*(0)=&{m_0}, \quad u^* (T)=0.
		\end{array}\right.
	\end{equation}
	From Eq.\eqref{fzeq}, it can be asserted that $u^* \ge 0$, otherwise it contradicts with $u^*(T)=0$. By Lemma \ref{comparison theorem}, we can further get $u^*(t)<k_{d}M_{max}^2T/w^2$. To sum up, it is known that $0 \le u^*<k_{d}M_{max}^2T/w^2$. So far, we have proved that the optimal control and state of problem \eqref{OC_c} must satisfy Eq.\eqref{fzeq}. By Theorem \ref{bvpeu}, we know that this solution exists and is unique. The desired result follows.
	\subsubsection{A Feedback Control}\label{analproof}
		Equation \eqref{PMP_optimal solution} can be transformed into a feedback control given by
		\begin{equation}
			\left\{\begin{array}{cl}
				\dot{M^*}(t)=&k_{a}M^*(t)(1-\frac{M^*(t)}{M_{max}})-k_{d}M^*(t)u^*(t) ,\\
				u^*(t)=&\frac{k_{a}(M_{max}-M^*(t))}{M_{max}k_{d}}\pm\sqrt{\frac{{M^*(T)}^2-{M^*(t)}^2}{w^2}+\left( \frac{k_{a}(M_{max}-M^*(t))}{M_{max}k_{d}}\right)^2},~t< T,
				\\M^*(0)=&{m_0}.
			\end{array}\right.
		\end{equation}
	
	First, we can derive from Eq. \eqref{PMP_optimal solution} that
	\begin{eqnarray*}
		&&\left[ k_{a}M^*(t)(1-\frac{M^*(t)}{M_{max}})-k_{d} u^*(t) M^*(t)\right]du^*+ \left[ \frac{k_{d}}{w^2}{M^*(t)}^{2}-\frac{k_{a}}{M_{max}}u^*(t)M^*(t)\right] dM^{*}=0,\\
		&&\left[ k_{a}(1-\frac{M^*(t)}{M_{max}})-k_{d} u^*(t)\right]du^*+ \left[ \frac{k_{d}}{w^2}{M^*(t)}-\frac{k_{a}}{M_{max}}u^*(t)\right] dM^{*}=0,\\
		&&d\left[  {M^*}^2-w^2{u^*}^2-\frac{2wk_{a}}{M_{max}k_{d}}M^*u^*+\frac{2wk_{a}}{k_{d}}u^*\right] =0,\\
	  &&{M^*}^2-w^2{u^*}^2-\frac{2wk_{a}}{M_{max}k_{d}}M^*u^*+\frac{2wk_{a}}{k_{d}}u^* =const.
	\end{eqnarray*}
	As ${M^*(T)}^2=const$ by $u^{*}(T)=0$, we have
	\begin{equation}
		{u^*}^2-2\frac{k_{a}(M_{max}-M^*)}{k_{d}M_{max}}u^*+\frac{{M^*(T)}^2-{M^*}^2}{w^2}=0,
	\end{equation}
	and
	\begin{equation}
			u^*(t)=\frac{k_{a}(M_{max}-M^*(t))}{M_{max}k_{d}}\pm\sqrt{\frac{{M^*(T)}^2-{M^*(t)}^2}{w^2}+\left( \frac{k_{a}(M_{max}-M^*(t))}{M_{max}k_{d}}\right)^2}.
	\end{equation}
	It is found that, when $M^*(t) \ll M_{max}$ is satisfied, "$+$" is initially chosen;  and then at the intermediate time $t_1$ which satisfies $\frac{{M^*(T)}^2-{M^*(t_1)}^2}{w^2}+\left( \frac{k_{a}(M_{max}-M^*(t_1))}{M_{max}k_{d}}\right)^2=0$, it changes to "$-$" until $T$. This choice will cause mass concentration of aggregates to decrease first and then increase.
Otherwise, we just choose "$-$" for the entire control process. This choice results in a monotonous increase in the mass concentration of aggregates.

	\color{black}
	
	\subsection{Derivation of the analytical solution}\label{Ap_A}
	
	The optimal control problem under the linear approximation is given by
	\begin{equation}\label{PMP_optimal solution2}
		\left\{\begin{array}{cl}
			\dot{M^*}(t)=&k_{a}M^*-k_{d} u^* M^*,\\
			\dot{\lambda^*}(t)=&-k_{a}\lambda^*+k_{d} u^*\lambda^* +2M^*,~t<T, \\
			M^*(0)=&{m_0}, \qquad \lambda^* (T)=0.
		\end{array}\right.
	\end{equation}
	Additionally, based on the optimality condition, $H(M^*(t),u^*(t),\lambda^*(t),t) \ge H(M^*(t),u(t),\lambda^*(t),t)$ with $H(M,u,\lambda,t)=\lambda (k_{a}M-k_{d} u M)-M^2-w^2u^2$,
	we have
	\begin{equation}
		u^*(t) \in {argmin}_{u\in [0,u_{max}]}
		\left[\lambda \left(k_{a}M^*-k_{d} u M^*\right)-(M^*)^2-w^2u^2 \right],
	\end{equation}
	or equivalently,
	\begin{equation}
\label{alge}
		0
		={\frac{\partial H(M,u,\lambda,t)}{\partial u}}\bigg|_{M^*,u^*,\lambda^*}
		=-2w^2u^*(t)-k_{d} \lambda^*(t) M^*(t).
	\end{equation}
	It is straightforward to solve above algebraic equation, which leads to \(u^*(t)=-{k_{d} \lambda^*(t) M^*(t)}/{2w^2}\) by noticing $w^2>0$. Combining this with the terminal co-state $\lambda^*(T)=0$ deduces \(u^* (T)=0\).
	
	In the next step, by taking the time derivative on both sides of the algebraic equation \eqref{alge},
	\[
	-\left(\frac{2w^2}{k_{d}}\right) \frac{du^*}{dt}=\lambda^* \frac{dM^*}{dt}+ M^*\frac{d\lambda^*}{dt},
	\]
 and substituting the ODEs of $M^*$ and $\lambda^*$ in Eqs. \eqref{PMP_optimal solution2} into the above one, we arrive at
	\begin{equation}\label{equ:du}
		\frac{du^*}{dt}=-\frac{k_{d}}{w^2}(M^*)^2.
	\end{equation}
By multiplying \(-2k_{d}M^*/{w^2}\) on both sides of the first equation in \eqref{PMP_optimal solution2} and substituting Eq. \eqref{equ:du}, a decoupled equation for the co-state variable $u^*(t)$ is derived, i.e.
	\begin{equation}
		\frac{d^2u^*}{dt^2}=2k_{a}\frac{du^*}{dt}-2k_{d}u^*\frac{du^*}{dt},
	\end{equation}
	which is of the second order. Integration once gives
	\[
	\frac{du^*}{dt}=-k_{d}{u^*}^2+2k_{a}u^*-C,
	\]
where $C$ is an undetermined constant. The above equation can be rewritten as
	\[\frac{du^*}{dt}=-k_{d}
\left[
\bigg(u^*-\frac{k_{a}}{k_{d}}\bigg)^2+\bigg(\frac{C}{k_{d}}-\frac{k_{a}^2}{k_{d}^2}
\bigg)\right],\]
whose solution gives the optimal trajectories of $u^*$ and $M^*$,
	\begin{equation*}
		u^*(t)=C_1\tan(-k_{d}C_1t+C_2)+\frac{k_{a}}{k_{d}},\quad
		M^*(t)=wC_1\sec(-k_{d}C_1t+C_2).
	\end{equation*}
	The constants \(C_1\) and \(C_2\) are determined by the boundary conditions \(M^*(0)={m_0}\) and \(u^* (T)=0\) as
	\begin{equation*}
		w^2C_1^2\sec^2(C_2)= {m_0}^2, \quad
		C_1^2\tan^2(-k_{d}C_1T+C_2)=\frac{k_{a}^2}{k_{d}^2},
	\end{equation*}
	which further gives
	\begin{equation*}
		u^*(0)=C_1\tan(C_2)+\frac{k_{a}}{k_{d}},\quad
		M^*(T)=wC_1\sec(-k_{d}C_1T+C_2).
	\end{equation*}
\subsection{$L_{1}$ Norm for the objective functional}\label{L1}
The quadratic cost in the main text is adopted to illustrate our conclusions, alternative forms of the objective functional are allowed too, such as the $L_{1}$ norm which reads
\begin{equation}
	J[u(\cdot)]=\int_{0}^{T}[M(t)+wu(t)]d t.
\end{equation}
To keep the non-negativity of $M(t)$ and $u(t)$, additional constraints are introduced too, that is
\begin{equation}
	\forall t \geq 0, \qquad M_{max} \geq M(t) \geq 0 \quad \text{ and} \quad u_{max}\geq u(t) \geq 0,
\end{equation}
here $u_{max}$ is not a fixed bound anymore. According to the optimality condition, $H(M^*(t),u^*(t),\lambda^*(t),t) \ge H(M^*(t),u(t),\lambda^*(t),t)$ with $H(M,u,\lambda,t)=\lambda (k_{a}M(1-M/M_{max})-k_{d} u M)-M-w u$,
we have
\begin{equation}
	u^*(t) =\begin{cases}
		u_{max}, \quad -(\lambda^* M^* k_{d} +w)>0,\\
		0, \quad -(\lambda^* M^* k_{d} +w)<0.
	\end{cases}
\end{equation}
The above description actually is a Bang-bang control.

So when do we have $(\lambda^* M^* k_{d} +w)=0$? It is found that this happens at only one time point. Based on PMP, we have
\[\frac{d \lambda^*(t)}{d t}=-k_{a}\lambda^*+2k_{a}\lambda^* M^*/M_{max}+k_{d} u^*\lambda^* +1, \lambda^*(T)=0\]
and
\begin{equation}\label{eqN}
	\frac{d N(t)}{d t}=(k_{a} N/M_{max}+1)M^*,
\end{equation}
where $N(t)=\lambda^*(t) M^*(t)$.
Since $N(T)=0$, we assert that $N(t)>-M_{max}/k_{a}$. It tells us an important fact that when $w/k_{d}>M_{max}/k_{a}$, there must be $u^*(t)=0, t\in [0,T]$. This states an extreme case when an inhibitor is too toxic and ineffective, we would better not to use it at all. Because of $(k_{a} N/M_{max}+1)M^*$, $N(t)$ increases monotonically to $0$, which means $-(\lambda^* M^* k_{d} +w)=0$ holds at one time point at most. Solving Eq.\eqref{eqN}, we get
\[ N(t)=\frac{M_{max}}{k_{a}}\left[\exp(k_{a}/M_{max}\int_{T}^{t}M^*(\tau)d \tau )-1\right]. \]

If the switch happens at $t_1$, it must satisfy the condition $N(t_{1})=-w/k_{d}$.
We have
\begin{equation}
	\frac{d M^*(t)}{d t}=
	\begin{cases}
		k_{a}M^*(t)\left(1-\frac{M^*(t)}{M_{max}}\right)-k_{d} u_{max} M^*(t), \quad t \in[0, t_1],\\
		k_{a}M^*(t)\left(1-\frac{M^*(t)}{M_{max}}\right),\quad t \in (t_{1},T].
	\end{cases}
\end{equation}
Through calculations, $t_{1}$ must satisfy
\begin{equation}\label{t1}
	(n_{1}-1)(\exp(k_{a}(T-t_{1}))-1)+(1/n_2-M_{max}/m_{0})\exp(-k_{a}n_{2}t_{1})-1/n_{2}=0,
\end{equation}
where $n_{1}=k_{d}M_{max}/(k_{a}w)$ and $n_{2}=1-k_{d}u_{max}/k_{a}$.
To sum up, the Bang-bang control is given by
\begin{equation}
	u^*(t) =\begin{cases}
		u_{max}, \quad t \in [0,t_1],\\
		0, \quad  t \in (t_1,T].
	\end{cases}
\end{equation}
where $t_{1}$ satisfies \eqref{t1}.

We have known that the optimal control trajectory is uniquely determined by the switch time $t_{1}$. How the switch time is affected by parameters $k_a$, $k_d$, $w$ and $u_{max}$ is of great importance. As shown in Figs. \ref{figure_7}(A1) and (C1), the switch time remains zero (meaning no inhibitor) when the apparent fibril growth rate \(k_{a}\) is either too large or too small, or when the rate constant for fibril inhibition \(k_{d}\) is very small. In the presence of large fibril inhibition rate \(k_{d}\) or small tolerance of inhibitors $w$, the switch time is close to 1 (meaning using inhibitors at every time), as long as $u_{max}$ remains small.
\begin{figure}[htbp]
	\centering
	\includegraphics[width=5in]{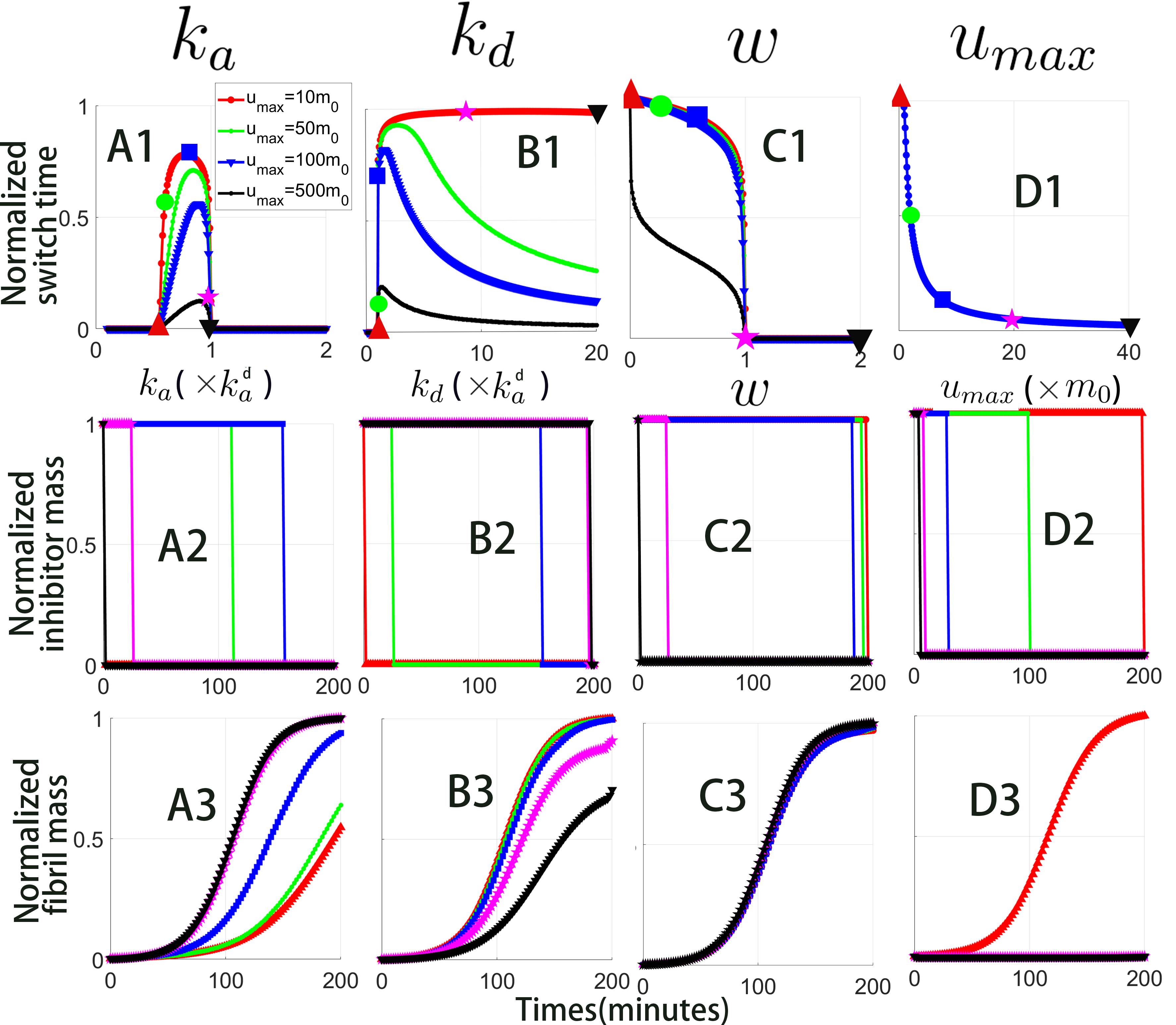}
	\centering
	\caption{Influence of model parameters on the Bang-bang control. Subplots in the first row show how the switch time depends on $k_a, k_d, w$ and $u_{max}$. The corresponding optimal trajectories for the inhibitor concentration and fibril concentration with parameters at five remarked positions are illustrated in the second and third rows.  Default parameters are set as \(m_0=2.7 \times 10^{-3}\)\(\mathbf{\mu M}\), \(k_{a}=5.6 \times 10^{-2} \min^{-1}\), \(M_{max}=1\)\(\mathbf{\mu M}\), \(k_{d}=k_{a}\), \(w=1\) and \(u_{max}=10m_{0}\) in figure(A-C). In figure(D), \(k_{d}=k_{a}/m_{0}\) is taken.}\label{figure_7}
\end{figure}

\color{black}
\subsection{Strategies for Drug Dosing}
\label{Ap_B}
The strategies used for drug dosing in the maintext are summarized as follows:

\paragraph{Lump-sum adding}
All inhibitors will be added at the initial instant as shown in Figs. \ref{figure_4}(A)-(B). The corresponding ODEs for $M(t)$ and $u(t)$ are
\begin{equation}\label{ycxtj}
	\left\{\begin{array}{cl}
		\dot{M}(t)=&k_{a}M(1-\frac{M}{M_{max}})-k_{d} u M,\\
		\dot{u}(t)=&-k_{d} u M,~t\leq T,
		\\M(0)=&{m_0}, \quad u (0)=u_{tot}.
	\end{array}\right.
\end{equation}

\paragraph{Constant adding}
Inhibitors will be added at a constant rate during the whole procedure from $t=0$ to $t=T$,
\begin{equation}\label{cztj}
	\left\{\begin{array}{cl}
		\dot{M}(t)=&k_{a}M(1-\frac{M}{M_{max}})-k_{d} u M,\\
		\dot{u}(t)=&\frac{u_{tot}}{T}-k_{d} u M,~t\leq T,\\M(0)=&{m_0}, \quad u (0)=\frac{u_{tot}}{T}.
	\end{array}\right.
\end{equation}

\paragraph{Periodic adding}

During the whole procedure, inhibitors are added intermittently. The adjacent adding of inhibitors and  non-adding forms a cycle.
Assume that there are \(N\) cycles $\{0,1,2,\cdots,N-1\}$.  Each cycle includes a time interval $(T/N-t_1)$ during which inhibitors are added uniformly, and a time interval $t_1$ when no inhibitor is added.
Let
\[u_{per}(t)=\left\{\begin{array}{cl}
	\frac{u_{tot}}{T-N t_1} & \quad \frac{kT}{N}\le t <\frac{(k+1)T}{N}-t_1,
	\quad k=0,1,2,...N-1,\\
	0&\quad\frac{(k+1)T }{N}-t_1 \le t <\frac{(k+1)T }{N},\quad k=0,1,2,...N-1.
\end{array}\right.\]
The following ODEs are obtained,
\begin{equation}\label{jgtj}
	\left\{\begin{array}{cl}
		\dot{M}(t)=&k_{a}M(1-\frac{M}{M_{max}})-k_{d} u M,\\\dot{u }(t)=&u_{per}-k_{d} u M,
		\\M(0)=&{m_0}, \quad u (0)=\frac{u_{per}}{T-N t_1}.
	\end{array}\right.
\end{equation}

\paragraph{Multiple times adding}

In this strategy, inhibitors are equally divided into $N$ pieces and each piece will be added sequentially into the solution after a constant time interval (see Figs. \ref{figure_4}(A)-(B) for example). Acccordingly, the optimal control is given by
\begin{equation}\label{mta}
	\left\{\begin{array}{cl}
		\dot{M}(t)=&k_{a}M(1-\frac{M}{M_{max}})-k_{d} u M,\\
		\dot{u}(t)=&\frac{u_{tot}}{N}\sum_{i=1}^N\delta(t-iT/N)-k_{d} u M,~t\leq T,\\M(0)=&{m_0}, \quad u (0)=\frac{u_{tot}}{N}.
	\end{array}\right.
\end{equation}

\subsection{Experimental Setup}

A$\beta40$ (purity $>98\%$) was purchased from ChinaPeptides. Other chemical reagents were bought from Aladdin. A$\beta40$ was pre-treated with ammonium hydroxide as previously described \cite{li2021Inhibitory}. A$\beta40$ was then solubilized in $6$ mM NaOH to $100 \mathbf{\mu M}$ while ThT and EGCG were directly solubilized in PBS. For ThT kinetics, the final concentrations of A$\beta40$ and ThT were $5 \mathbf{\mu M}$ and $20 \mathbf{\mu M}$, respectively. Corning $3603$ $96$-well plates were used and ThT kinetics was monitored at a Tecan Infinite M200 PRO microplate reader with continuously shaking. Fluorescence intensities at ex = $440$ nm and em = $480$ nm were recorded every $6$ min.

For  the ``Lump-sum'' group, $10 \mathbf{\mu M}$ EGCG was mixed with A$\beta40$ and ThT at $t=0$ min.

For the ``Twice adding'' group, $5~ \mathbf{\mu M}$ EGCG was mixed with A$\beta40$ and ThT at $t=0$ min. Then stock solutions of EGCG ($0.5 m\mathbf{M}$) was added to the mixture at $t=30$ min or $t=60$ min to increase the concentration of EGCG to $10~\mathbf{\mu M}$ while minimize the changes of A$\beta40$ concentration.

For the ``Four times adding'' group, $2.5~ \mathbf{\mu M}$ EGCG was mixed with A$\beta40$ and ThT at $t=0$ min. Then stock solutions of EGCG ($0.5 m\mathbf{M}$) was added to the mixture at $t=30, 60, 90$ min or $t=60, 120, 180$ min to increase the concentration of EGCG to $5, 7.5$ and $10~\mathbf{\mu M}$, respectively.

\subsection*{Declaration of Competing Interest}
The authors declared no competing interests.

\subsection*{Acknowledgements}
The authors acknowledged the financial supports from the National Natural Science Foundation of China (Grants No. 21877070, 22007044, 12205135), the Natural Science Foundation of Fujian Province of China (2020J05172, 2020J01864), and the Startup Funding of Minjiang University (mjy19033).

\bibliographystyle{unsrt}
\bibliography{bibfile}

\end{document}